\documentclass[times, 10pt]{article}
\usepackage{latex8}
\usepackage{verbatim}
\usepackage{graphicx}
\usepackage{subfigure}
\usepackage[lined,boxed,commentsnumbered, ruled]{algorithm2e}
\usepackage{amsmath, amsthm, amssymb}
\newtheorem{thm}{Theorem}[section]
\newtheorem{cor}[thm]{Corollary}
\newtheorem{lem}[thm]{Lemma}

\begin{document}

\title{Approximate Discovery of Service Nodes by Duplicate Detection in Flows}

\author{Zhou Changling$^{1,2}$  Xiao Jianguo$^{3}$  Cui Jian$^{1}$  Zhang Bei$^{1}$  Li Feng$^{4}$\\
$^1$ Computer Center of Peking Univeristy\\
$^2$ School of Electronics Engineering and Computer Science of Peking University\\
$^3$ Institute of Computer Science \& Technology  of Peking University
 \\ Beijing ,100871 ,P. R. China \\
$^4$ ShenYang Academy of Governance,  ShenYang, 110000,P. R. China
\\
}

\maketitle

\begin{abstract}

%

Knowledge about which nodes provide services 
is of critical importance for network administrators.
Discovery of service nodes can be done by making full use of duplicate element detection in flows.
Because the amount of traffic across network is massive, especially in large ISPs or campus networks,
we propose an approximate algorithm with Round-robin Buddy Bloom Filters(RBBF)
for service detection using NetFlow data solely.
The properties and analysis of RBBF data structure are also given.
Our method has better time/space efficiency than conventional algorithm with a small false positive rate.
We also demonstrate the contributions through a prototype system by real world case studies.

\end{abstract}

\section{Introduction}

Today's computer networks support many different types of services.
The knowledge about which nodes provide the services 
is of critical importance
for anyone who is responsible for managing and protecting organization's network.
For example, undesired services, which  often result from compromised hosts, break up network policy
and cause vulnerability and information disclosure.
Monitoring the availability of services can help us to find network failure or hosts down events.
Meanwhile, there is an inevitable trend that IPv6, as the preferred next generation Internet protocol,
is widely adopted and used  with rapid development of Internet and exhaustion of IPv4 addresses.
But there are not enough IPv6-compatible networking devices deployed.
Without the protection like in IPv4 environment,
IPv6 service providers have more difficulty in applying 
security policies and are more vulnerable to intrusions and attacks.
The knowledge of IPv6 services is still far from enough without the help of IPv6-ready devices.
Discovering the service providers and paying more attention to them is a key step to gain 
awareness of network.

We call a service provider a service node(SN), which receives request messages from clients 
and generates response messages.
A service node is hosted on a computer, identified by its IPv4/IPv6 address and
a specific transport layer port, thus a unique end node tuple \{ip address, port number, protocol\}
is used for every service node.  
For example, a web server www.example.com \cite{rfc2606}  provides service on "http://192.0.43.10:80" or "http://[2001:500:88:200::10]:80".
In this case the ip address is 192.0.43.10 or 2001:500:88:200::10, the port number is 80 and the protocol is TCP.
A service node provides service repeatedly for some clients or a single client several times in a fixed period.
In this paper,we don't distinguish between transient and permanent service nodes, and  focus only on TCP and UDP protocols.
Specifically, we view P2P transactions as following the client/server model,
even if the server node in this case may handle only a limited number of clients for a limited amount of time.

We use passive discovery of service nodes from NetFlow data.
Unlike active probing which actively scans the network, passive service discovery uses traffic information\cite{Bartlett2007}.
In large organizations networks such as campus networks and some ISPs, the amount of traffic across the network
is too large to capture and see every packet's
payload,  which makes packet level service discovery\cite{pads} impractical.
In those cases, the most widely used information source is NetFlow, which uses tuples like \{source ip address, source port number, destination ip address,
destination port, protocol\} as a key, where a flow is defined as a unidirectional sequence of packets that all share the same key.
By this way, NetFlow offers a good trade-off between the level of details provided and scalability.
A majority of networks have been equipped with devices to collect and export NetFlow,
and many tools are available to process such data \cite{nfdump,ntop,SiLK}.
NetFlow gains more and more attention from network and security administrators and researchers\cite{FloMA,So-In2009,Lakkaraju04nvisionip:netflow,Barford2002a,Moore2005a}.

The foundation of our method, passive discovery of service nodes from NetFlow data,
is detecting duplication in flows.
The basic idea comes from the observation that service nodes usually
perform similar interactions with several other hosts, usually called clients, over a certain time period.
This assumption, is confirmed by behavioral analysis of network traffic,
such as BLINC proposed by Faloutsos
et al. \cite{Karagiannis2005a}
and traffic activity graph decomposition proposed by Jin et al. \cite{Jin2009}.
A service-based traffic classification using that idea is proposed by Baldi et al.\cite{Baldi2009},
but their service discovery module, which they call it service identification,  is a
payload-based implementation.
A work that is very close to ours was done by Berthier et al.\cite{Berthier2010}
using NetFlow for passive service detection.
They detect duplicate service node tuples by directly tracking all tuples in flows, i.e., by maintaining a counter for each tuple.
They use a set of heuristics rules to carry out server identification, such flow timing, port number defined by
IANA\cite{INNA}, etc. These heuristics rules are later combined by using naive Bayesian method.
Their method is somehow complicated and has relatively low efficiency. 



To effectively detect duplication for streaming data,
several approximate algorithms are proposed.  Metwally  et al. \cite{Metwally2005} 
propose to use \emph{counting Bloom filters} when detecting click fraud in
pay-per-click advertisement systems.
A counting Bloom Filter supports both inserting and deleting elements to sets.
If a fresh element is inserted, the corresponding counters will be increased;
conversely, the corresponding counters will be decreased if a stale element is deleted.
However, this method is effective only when the stale element to be removed is known,
which is impossible when there is insufficient memory to keep the exact element sequence.
Deng and Rafiei  \cite{Deng2006} 
introduce  \emph{stable Bloom Filters}  to randomly discard stale elements in landmark windows.
The stable Bloom Filter sets cells corresponding to input data to the maximum value
and decreases the values of randomly selected cells whenever data arrives. However, Stable Bloom Filters
have false positive errors as well as false negative errors.
In addition, the parameter configuration in Stable Bloom Filters is complex and
its solution is not always optimal.
Wang et al. \cite{Wang2010} 
use \emph{decaying bloom filter} to detect duplicates for data streams in a sliding window.
It sets cells corresponding to input data to the maximum value similar to \emph{stable Bloom filters},
but decreases all the non-zero counters by 1 before inserting a new element.
The problem  with their approach is that the time complexity of inserting an element is
much higher than others and is unacceptable for many real-time applications.
Lee, et al. \cite{Lee2011} propose \emph{Time Interval Bloom Filter} to eliminate duplicate
RFID data in a small time interval. Their method is applicable only for time sensitive data.

Although those approximate algorithms based on variations of Bloom Filters  have better performance in the sense
of space/time trade-off, they cannot be directly applied to our  application because of the following problems:



Firstly, the elements of interest are different from duplicate detection algorithms listed above.
Our goal of discovering service nodes  in network traffic flows is to select and remember duplicate elements,
but traditional duplication detect algorithms prefer to eliminate them.

Secondly, a network session is a valid communication between one client-end node and one server-end node.
A real world network session lasts from milliseconds to several hours.
By contrast NetFlow collection process saves data to files at a certain interval, which is usually 5 minutes.
Thus, NetFlow may break up a logical flow into multiple flows.
This type of repetition is an artificial duplication and  should be eliminated.

Thirdly, a valid network transaction is always bidirectional,
thus flows reported should always be bidirectional.
Flows consists of two end node tuples,
one is source side and the other is destination side. When bidirectional flows are captured from network,
the two flows with matching end node tuples, source and destination exchanged, should be taken as 'duplicate'.
Because obviously we cannot tell which side of a flow
is the repeatedly used service node in advance,
both source part and destination part must be checked.
If we check all flows, surely the same \{ip address, port number\} pairs will be seen at least twice in flows of both direction.


%
Thus the task of accurately detecting servers based solely on NetFlow flows is challenging.
We propose a new data structure, \emph{Buddy Bloom Filters} with round-robin schema
and a two stage algorithm to solve the problems listed above.
The contributions of this paper include:  
\begin{itemize}

\item The structure and properties of \emph{Buddy Bloom Filter(BBF)} are given.
To mitigate boundary effect of a single BBF, multiple BBFs are used in round-robin schema
over jumping windows. The feature that bitwise OR operation can be used over multiple BBFs
is further exploited to speed up the query on history windows.


\item Analysis of the data structure is given, including false positive rate, parameters selection,
computational complexity, etc.
A two stage algorithm based on \emph{Round-robin Buddy Bloom Filter(RBBF)} is proposed,
which is specifically used for detecting service nodes form NetFlow data.
Besides, the storage schema of detecting results is proposed,
which can be extended to support rapid query of existence operation.

\item The performance of proposed algorithm in comparison with hash table is analyzed.
A prototype system,
which is compatible with both IPv4 and IPv6,
using the proposed data structure and algorithm is introduced.
Some real world use cases of the system are discussed.


\end{itemize}


The rest of this paper is organized as follows. Section 2 reviews some background knowledge of bloom filter
and proposes a novel data structure \emph{Buddy Bloom Filter} along with the round-robin schema.
The properties of \emph{RBBF} and analysis of \emph{RBBF} are also included in this section.
In Section 3, we present the service nodes detection algorithm based on \emph{RBBF}.
We report the results of our experiments and discuss a number of use cases in Section 4,
followed by some concluding remarks in Section 5.

\section{the RBBF Data Structure}

The classical Bloom Filter was proposed by Burton Bloom in 1970 \cite{Bloom1970}
and used for detecting approximate membership of elements.
A Bloom Filter uses k hash function $h_1, h_2 ,\dots,h_k$ to hash each element of a set $F$ into a bit array of the  size $m$,
where $F$ comes from a universe $U$.
All bits in Bloom Filter are initially set to 0.
For each element $f \in F$, the insertion process marks some bits in Bloom Filter to 1,
which positions are indexed by  $h_1(f),h_2(f),\dots, h_k(f)$.

Classical Bloom Filter allows membership queries. For example, given a newly arrived element $f \in F$,
whether $f$ is a duplicate in $F$ can be determined by the bits at the positions
$h_1(f), h_2(f), \dots, h_k(f)$.
If any of these bits is zero, we know $f$ is a distinct fresh element.
Otherwise it is regarded as a duplicate with a certain probability of error.
This method has a small probability of producing a false-positive error,
i.e., a distinct element wrongly reported as duplicate.
If the hash functions are perfectly random, the probability of a false positive (false positive rate).
$Pr(FP) = (1-P_0)^k = (1-(1-\frac{1}{m})^{kn})^k  \approx (1-e^{-\frac{kn}{m}})^k$,
where $P_0$ is the probability that a specific bit is still 0 after inserting n distinct elements.

Although Bloom Filter is simple and space efficient, it does not allow deletions.
Deleting elements from a Bloom Filter cannot be done simply by changing the corresponding bits back to zeros,
as a single bit may correspond to multiple elements.
Therefore, in the continually incoming NetFlow flows, which is a typical data stream environment,
if we apply a single Bloom Filter in detecting duplication, when more and more new flows arrive,
the portion of zeros in Bloom Filter decreases continuously, the false positive rate increases accordingly,
and finally reach the limit one. At this time every distinct element is reported as duplicate,
indicating that Bloom Filter fails completely, a state of Bloom Filter called "full".

\subsection{Buddy Bloom Filter}

We define a pair of Bloom Filters that share the same hash functions a Buddy Bloom Filter(BBF).
Each BBF is composed of two bit arrays B1 and B2, of the same size $m$.
A newly arrived element $f_i$ is query-then-set on bit array B1 first,
If it is the first time we see this element, only the bits positioned by hash functions will be set in B1.
If it is a repeatedly occurring element, an extra following step of setting on bit array B2 is performed.
The state change of BBF is given in Figure \ref{fig:insert}.
A select-then-remember schema of duplicate elements is  actually carried out.
The bit array B2 of BBF represents those duplicate elements.
We call B1 the selecting part and B2 the remembering part of a single BBF.
In practice, for efficiency consideration,
when an element comes, the first step is to query B2. If existence estimation is true,
the element is skipped and no more action is taken. Otherwise query-then-set on B1 is followed ,
then set on B2 conditionally.


\begin{figure}[htbp]
\begin{center}
  \includegraphics[width=6.0in]{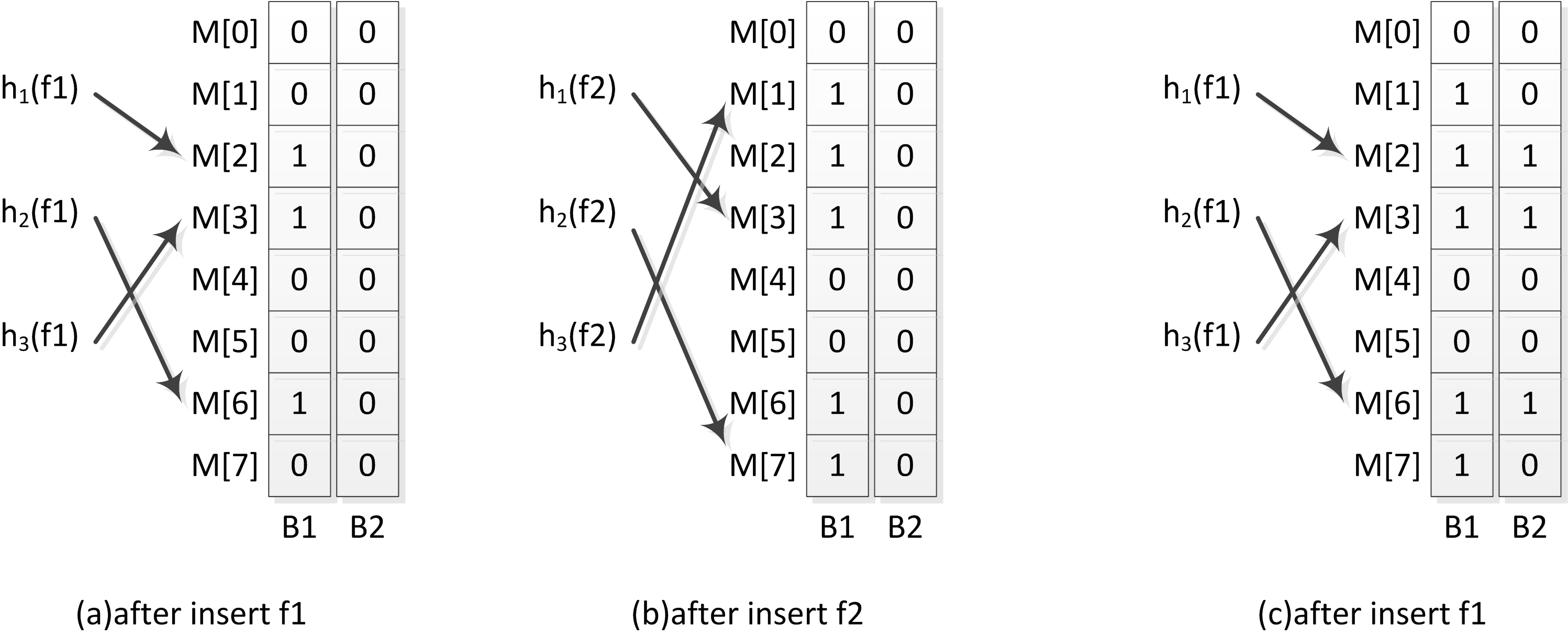}
  \caption{State of Buddy Bloom Filter}\label{fig:insert}
\end{center}
\end{figure}

A single Buddy Bloom Filter is an extension of classical Bloom Filter.
It doubles memory requirement like counting Bloom Filter\cite{Broder2004} with two-bit counter,
but they are different in structure and have different properties.
A counting Bloom filter uses counters instead of single bits as representing cells, and
each counter records the number of elements that are
associated with the corresponding cell. If a fresh element is
inserted, the corresponding counters are increased;
conversely, the corresponding counters are decreased if
a stale element is deleted.
A Buddy Bloom Filter does not support deletion of elements.
The two-bit counter of counting Bloom Filter is an indivisible cell, which can count number $\{0 \dots 3\}$.
On the contrary, the arrays of Buddy Bloom Filter can be regarded as two side-by-side buddies,
the selecting part B1 and remembering part B2 are actually two Bloom Filters with the same group of hash functions.
Each index position $i$ of Buddy Bloom Filter $B1[i]$ and $B2[i]$ can be set individually in theory.

\subsection{Round-robin Buddy Bloom Filter}

To allow stale elements to be removed from sets, we use time windows to separate elements.
A naive approach is that every window correspond to a Buddy Bloom Filter.
However, this approach introduces boundary effect, as shown in Figure\ref{fig:boundary effect}.
In NetFlow data processing, a flow $f_i$ is unidirectional. Therefore valid communications are always bidirectional,
then the conjugate of original flow $\bar{f_i}$, with source and destination swapped, is expected to appear in the set.
If a flow $f_i$ arrives at time $t_1$ and its matching flow $\bar{f_i}$ arrives at time $t_2$,
since $f_i$ and $\bar{f_i}$ is recorded in different time windows, they are deemed as unmatched unidirectional flow.
Unprocessed boundary effect causes false negative report.

\begin{figure}[htbp]
\begin{center}
  \includegraphics[width=3.0in]{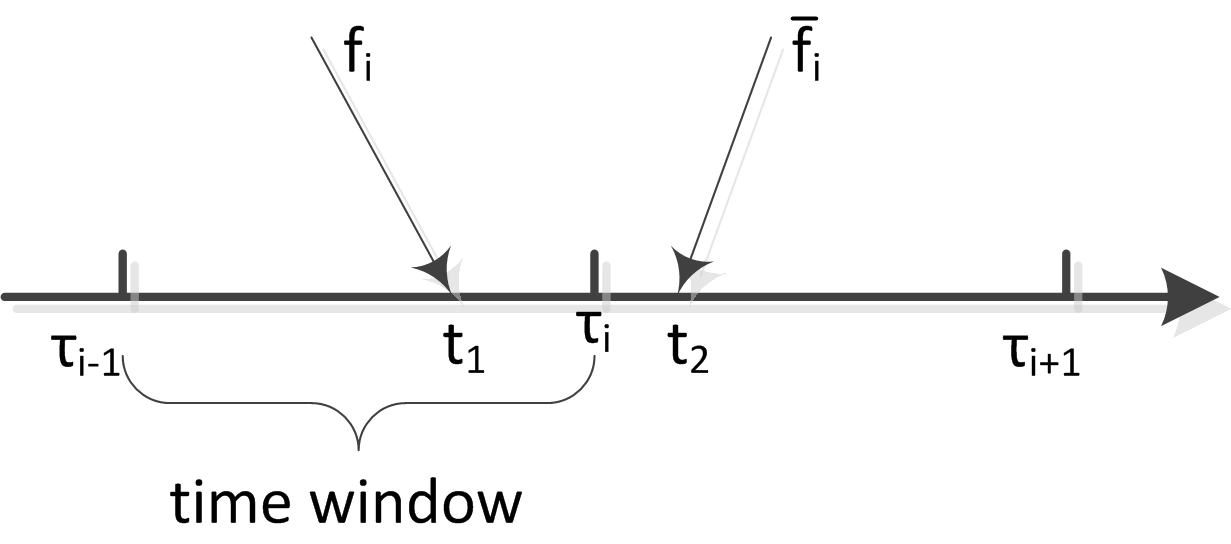}
  \caption{Boundary Effect of Single Buddy Bloom Filter}\label{fig:boundary effect}
\end{center}
\end{figure}

To reduce the boundary effect, we partition a discrete-time interval $\Gamma$ of into $M$ successive time windows,
where $M$ is the number of windows big enough to mitigate the boundary effect.
Assume the length of a window is $\gamma$. Then, we have $\Gamma = M \times \gamma $.
Each time window uses a Buddy Bloom Filter. Among these BBFs, $M-1$ BBFs are in the historical stage of
monitoring interval, one BBF is currently coping with incoming elements.
The query of existence is operated in the entire time interval $\Gamma$.
A round-robin strategy is deployed with these BBFs.
Before jumping to a new time window, the oldest BBF is deemed as stale and will be cleared and
reused in the new time window. The RBBF schema is shown in Figure \ref{fig:RBBF}.

\begin{figure}[htbp]
\begin{center}
  \includegraphics[width=4.0in]{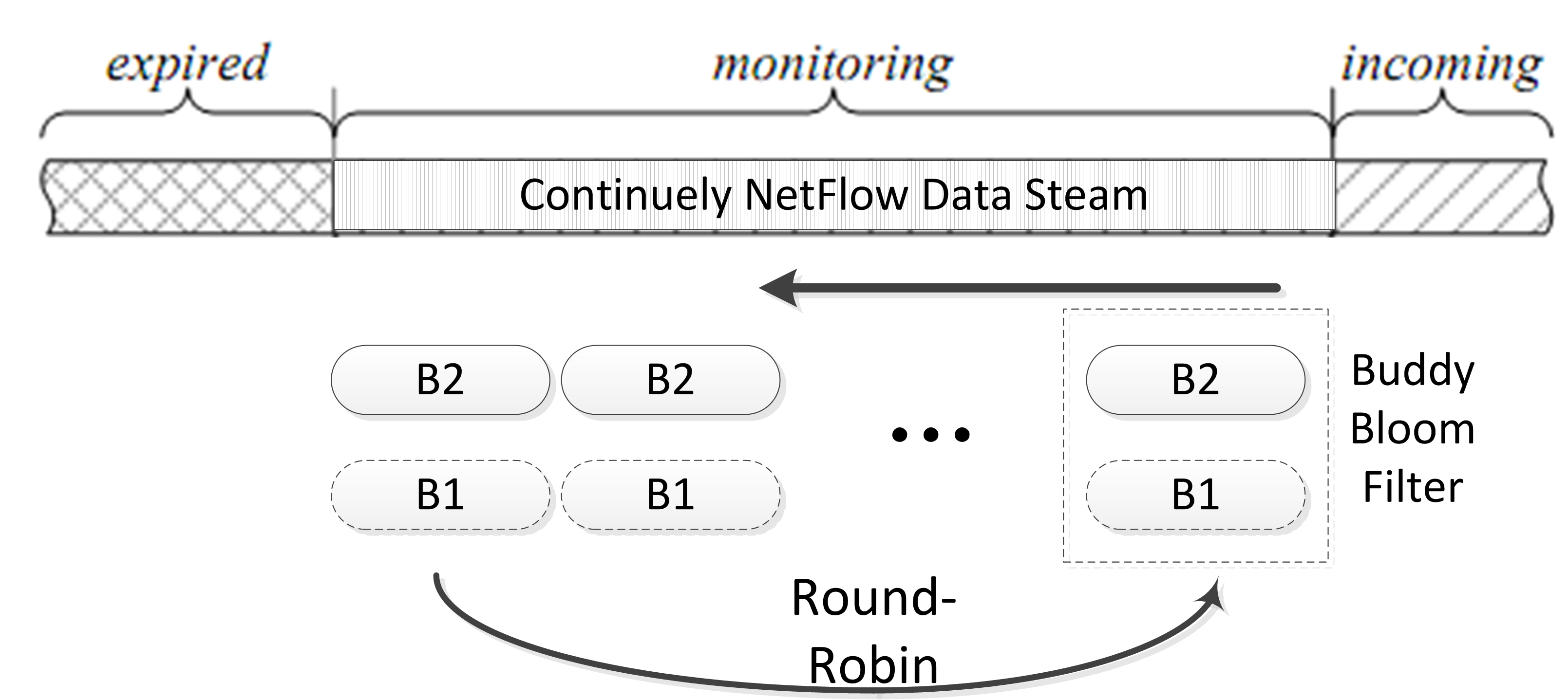}
  \caption{Round-robin of Buddy Bloom Filters}\label{fig:RBBF}
\end{center}
\end{figure}

The basic idea behind round-robin windows usage of Buddy Bloom Filter is jumping window model\cite{Zhu2002},
The jumping window model divides a data stream into smaller disjoint sub-windows,
where each sub-window  corresponds to a fixed time interval or contains a fixed number of elements.
A jumping window covers a certain number of sub-windows and slides in jumps as the data flows.
Once all of the sub-windows are filled, the oldest sub-window is removed,
and a  new sub-window is started to accommodate fresh elements.
The data structure proposed here uses a single Buddy Bloom Filter to associate with each sub-windows.
Specially for NetFlow data processing, sub-window width can be chosen as the same time interval as
the period of collecting and saving data files, which is usually 5 minutes.



One great feature of Buddy Bloom Filter is that BBF support bitwise operation on its cell, which is shown in Figure \ref{fig:bitwise summary}.
When multiple Buddy Bloom Filters are used,
the query on history windows can be done
on summary BBF by only one comparison.

\begin{figure}[htbp]
\begin{center}
  \includegraphics[width=3.0in]{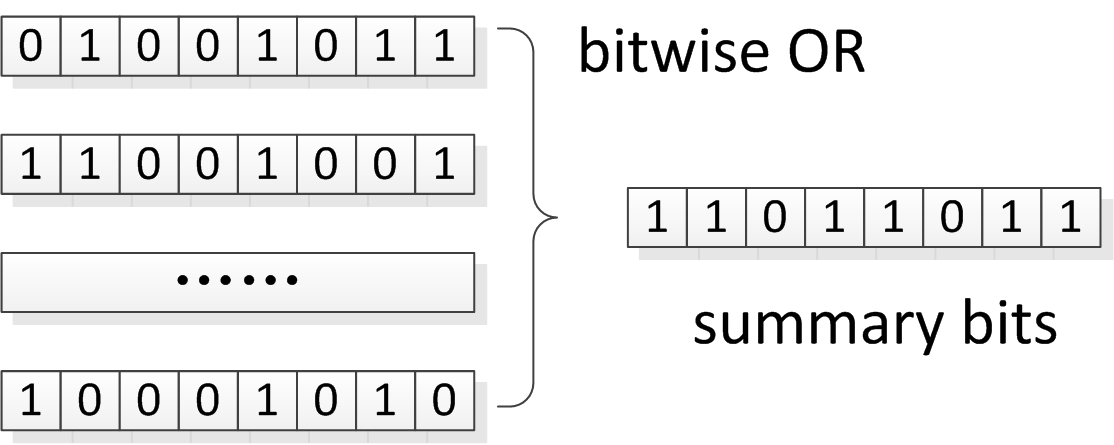}
  \caption{Bitwise Summary of History Windows}\label{fig:bitwise summary}
\end{center}
\end{figure}



\subsection{Analysis of Round-robin Buddy Bloom Filters}


Firstly let's consider using a single Buddy Bloom filter (BBF) with two arrays that each has $m$ bits
and a group of $k$ hash functions to detect duplicate items
in a static set of $N$ elements that has a ratio of $r_d$ distinct elements.

\begin{lem}\label{BBF false positive rate}
There could be only false positives in single BBF, and the probability of false positive (false positive
rate ($FP_{BBF}$)) is

\begin{equation}\label{eq:bbf-fp}
    PR(FP_{BBF}) \le (1 - e^{ -{kN*(1-r_d-(2r_d-1)*(1-e^{-{kNr_d}/{m}})^k )}/{m}})^k
\end{equation}

\end{lem}

\begin{proof}
A Buddy Bloom Filter can be viewed as two separated Bloom Filters.
B1 is used for elements selection, and B2 is used to remember duplicate elements.
Thus there can be only false positives in a single BBF.

As stated in\cite{Broder2004}, the probability of a false positive of a single Bloom Filter ($FP_{BF}$)  is
$$PR(FP_{BF}) = (1-P_0)^k = (1-(1-\frac{1}{m})^{kn})^k  \approx (1-e^{-\frac{kn}{m}})^k$$,
where $P_0$ is the probability that a specific bit is still 0 after inserting n distinct elements.

Suppose the duplicate elements are evenly distributed. There are $N(1-r_d)$ elements that appear more than one time.
These duplicate elements are sent to B2 and remembered for the first time.
There are no more than $N(1-r_d)$ actually duplicate elements in B2, since some elements can appear more than twice.
Also some elements can be sent to B2 wrongly, according to B1's false positive rate.
The number of elements appear only once in B1 is $N(2r_d-1)$, each has a false positive rate no larger than $(1-e^{-{kNr_d}/{m}})^k$.
Thus, totally there will be less than $N*(1-r_d) + N*(2r_d-1)*(1-e^{-{kNr_d}/{m}})^k$ distinct elements in B2.

Substitute the parameter of original Bloom Filter false positive rate, and then we get the answer.
\end{proof}

Thus Buddy Bloom Filter has no false negatives similar to original Bloom Filter.
With regards to network traffic processing, especially for the algorithm proposed in section 3,
this means no duplicate elements of interest are ignored,
with a chance of bringing some not really duplicate elements into the set, increasing the false positive rate.

One measures of determining how efficient Bloom Filters are is to consider how
many bits $k$ are necessary to represent all sets of n elements from a universe in
a manner that allows false positives for at most a fraction $\varepsilon$ of the universe but
allows no false negatives.

\begin{cor}\label{simple_fp_b1_b2}
To construct a Buddy Bloom Filter for a static set with known cardinality $n$,
if $FP_{BBF}$ must  be restricted to a threshold  $\varepsilon$ with minimal space overhead,
the optimal parameters are

\begin{equation}\label{eqn:k m opt}
\begin{split}
 k_{opt} &= \lceil log_2(1/\varepsilon) \rceil \\
 m_{min} &=  \lceil n r_d \cdot  log_2e  \cdot log_2{(1/\varepsilon)}  \rceil
 \end{split}
\end{equation}

\end{cor}

\begin{proof}
A false positive occurs if all the  $k$  bits
corresponding to the item $x$ being queried are occasionally set
to $1$ by some of the $n$ elements when $x$ is not in the set,
It has been proven\cite{Broder2004}  that false positive probability of a Bloom Filter
$PR(FP_{BF})$ can be minimized to
$(1/2)^{ ln2 \cdot (m/n)}$
when $k =ln2  \cdot (m/n)$.

To construct a single  Bloom Filter of the given conditions,  
the optimal number of hash functions is $k_{opt} = \lceil log_2(1/\varepsilon) \rceil$
and the required minimal space (in bits) is $m_{min} = \lceil log_2e \cdot k_{opt} \cdot n \rceil $.
Clearly,  $FP_{BF}$  does not reach $\varepsilon$  until all the  $n$  elements are
inserted, thus  $n$  is also  called the designed capacity of a BF and $\varepsilon$ is
also referred to as the target error rate, denoted as $F$ .

With Buddy Bloom Filter, there are two BFs with different cardinal.
The simple rule of thumb of relationship is that, each element in B2 can be found in B1, the cardinal of B2 is smaller than B1.
And B1 and B2 have also the same bit array length and same hash function groups. Thus, $B2 \subset B1$.
So that $ FP_{B1} > FP_{B2} $.
The optimized parameter of k is no larger than the $k_{opt}$ of B1.
For the same reason, we can induce the minimum length of bit array of each BF in BBF.
Total number of bits needed for a BBF is $2*m_{min}$.

\end{proof}


False positives can also occur when we do a union of historical sub-window duplication sets by
taking the OR of the bit vectors of the original Bloom Filters to generate a
summary representation.
Broder etc. \cite{Broder2004} 
give an estimated error rate of two sets union operation.
Suppose that one has two Bloom Filters representing sets S1 and S2
with the same number of bits and using the same hash functions. Intuitively,
the inner product of the two bit vectors is a measure of their similarity. More
formally, the $jth$ bit will be set in both Filters if it is set by some element in
$S1 \bigcap S2$, or if it is set simultaneously by some element in $S1 - (S1 \bigcap S2)$
and by another element in $S2 - (S1 \bigcap S2)$.
The probability that the jth bit is set in both Filters is therefore
$ m(1-(1-1/m)^{k|S_1|} - (1-1/m)^{k|S_2| + (1-1/m)^{k(|S_1|+|S_2|-|S_1 \cap S_2|)}}$
given $|S_1|, |S_2|, k, m$, and the magnitude of the inner product, one can
calculate an estimation of the intersection $|S1 \cap S2|$ using the equation above.
Similar operation can be applied to more than two sets union together.
Luckily for network traffic flows processing, especially for the algorithm stated in Section 3,
the final result is not different from considering the multiple sub-windows as a whole one.
and false positive is the same as a single Buddy Bloom Filter that spans N sub-windows.


\begin{cor}\label{RBBF false positive rate}
There could be only false positives in RBBF, and the probability of false positive (false positive
rate (FP)) is
\begin{equation}\label{eq:rbbf fpr}
 Pr\{FP_{RBBF}\}  \approx (1 - e^{ -{kMN*(1-r_d-(2r_d-1)*(1-e^{-{kMNr_d}/{m}})^k )}/{m}})^k \\
\le (1 -  e^{-kMNr_d/m})^k
\end{equation}
\end{cor}

\begin{proof}

Suppose that Round-robin Buddy Bloom Filters are composed by $M$ single BBF, each BBF with two array of $m$ bits
and a group of $k$ hash functions,  is used to detect the duplicate items of a live data stream
with each window covering on average of $N$ elements and the ratio of distinct elements is $r_d$.
If we treat the $M$ sub-windows as a whole, the result is obvious.


\end{proof}

\begin{thm}
(Space Requirement). Given that $m$ and $k$ are constants, processing each data stream element needs O(1) time,
 independent of the size of the space and the stream.
\end{thm}

\begin{proof}
From the discussion above, we already know that the space that this data structure need is totally $2*(M+1)$ bit arrays.
Of these bit arrays, one is used for the summary BBF. Except for the current time window used in BBF,
the other $M-1$ BBF is used in an round-robin strategy.  Every BBF has 2 bit arrays, so a total of $2*(M+1)$ space is needed.
\end{proof}


Similar with original Bloom Filter, the time complexity of RBBF are constants.

\begin{thm}
(Time Complexity). Given that $m$ and $k$ are constants, processing each data stream element needs O(1) time,
 independent of the size of the space and the stream.
\end{thm}

\begin{proof}
As for original Bloom Filter,
if $m$,$k$ are constants,  and for the specified hash functions,
the time usage of preparing and processing data is constant,
which means there has no impact on the processing time.
So, the time complexity is $O(1)$.

\end{proof}

\section{Service Nodes Detection Algorithm}

The procedure of service nodes discovery can be summarized as Figure \ref{fig:cbbf-detect}.
Flows come and are processed by Algorithm \ref{alg:DetectDupFlow}, the flow-level duplicate detection.
The conjugate flows are remembered and their tuples are divided into two parts,
\{source address, source port, protocol\} and \{destination address, destination port, protocol\}.
Each part can be viewed as a candidate of service nodes. Then the end-node level duplicate detection by Algorithm \ref{alg:DetectDupNode} is
performed, the repeated occurring tuples \{ip, port, proto\} are of concern.
The procedure is shown in Figure \ref{fig:cbbf-detect}.

\begin{figure}[htbp]
\begin{center}
  \includegraphics[width=3.0in]{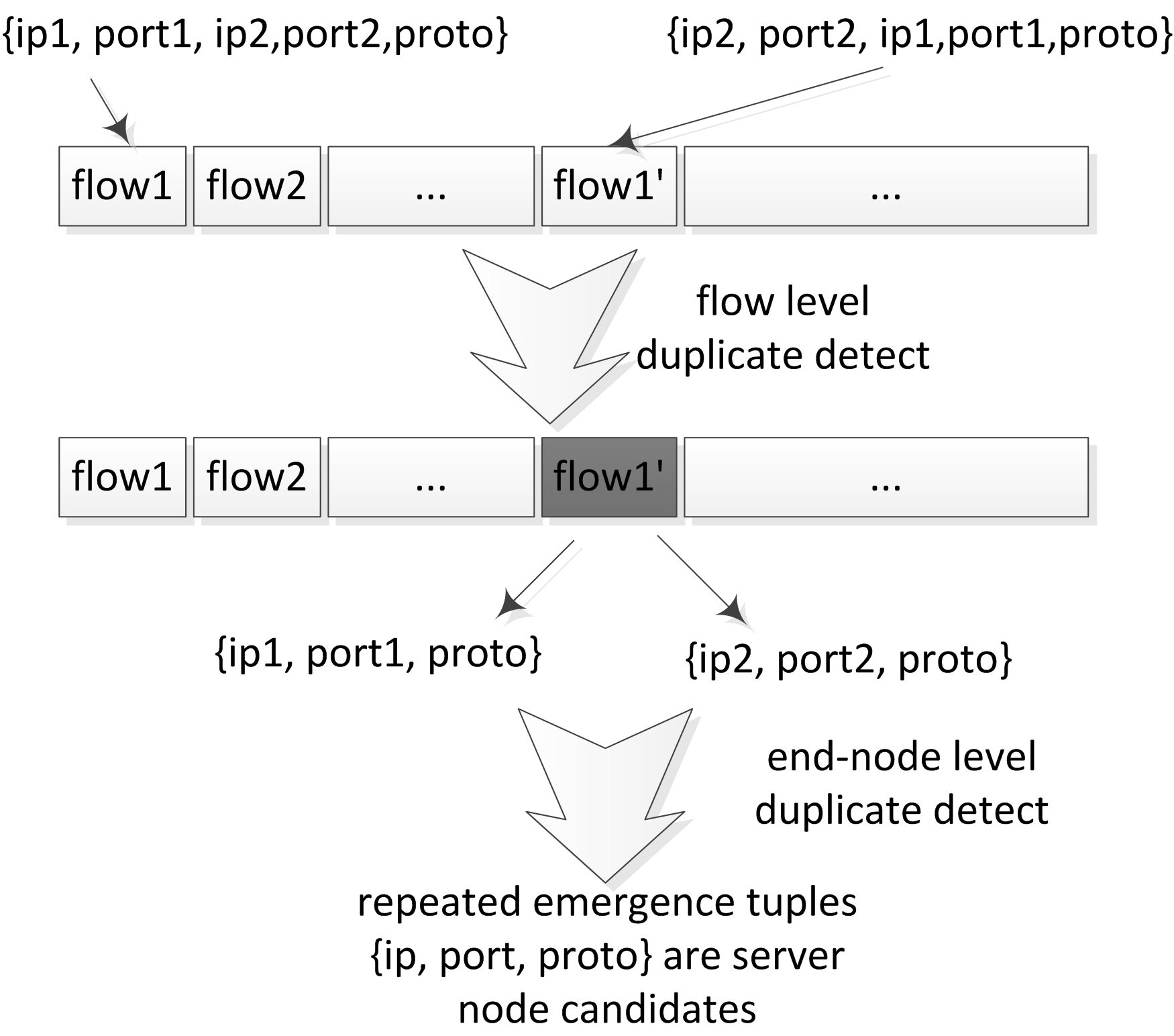}
  \caption{Procedure of service nodes Discovery}\label{fig:cbbf-detect}
\end{center}
\end{figure}

Firstly we must eliminate the unwanted duplication in traffic flows.
As we have mentioned before, NetFlow may break a logic flow into separate flows,
the end node tuple can be seen twice in bidirectional flow.
Every unidirectional flow is formed by different combinations of key
\{source ip address, source port number, protocol, destination ip address, destination port number\}.
Thus if an element of flow key is hashed by $k$ hash functions using the unique key combination,
$k$ bits in a bit array will be chosen according to the hash functions $h_i(x) (1 \leq i \leq k)$.
The position of selected bits are prepared for query or insertion operation.

\begin{algorithm}
\caption{Round-robin schema of  Buddy Bloom Filters(RBBF)}
\label{alg:RoundrobinBuddyBloomFilter}
\SetKwData{Index}{Index}
\SetAlgoNoLine
\LinesNumbered
\KwIn{$N$ Buddy Bloom Filters $BBF[1] \dots BBF[N]$, index of stale Buddy Bloom Filter $ind$}
\KwOut{a summary Buddy Bloom Filter of history windows}
\BlankLine
Initialize $BBF_{OUT} = 0$ \\
\For{ $j \leftarrow (1 \dots N) \cap j \neq ind $} {
    $BBF_{OUT} \longleftarrow BBF_{OUT} \bigvee BBF[j]$
}
\Return{$BBF_{OUT}$}
\end{algorithm}

When querying in the element, the proposed algorithm here permits only the element seen for the first time to be processed further.
During inserting process, flow keys of both directions are inserted into the data structure,
to eliminate the effect of artificial duplication in bidirectional flow,
by exchanging the source and destination parts accordingly.
Here we notate the after exchanging flow key by $\bar{f_i}$, conjugates with original flow key $f_i$.
Long alive flows are detected using round-robin strategy of Bloom Filters, see algorithm
\ref{alg:RoundrobinBuddyBloomFilter}.
The number of bit arrays needed for Bloom Filters
$N$ is set according to the NetFlow collection period $t_0$ ( usually 5 minutes) and the
flow time out period of NetFlow $T_{timeout}$( usually 15 minutes ), $N = \lceil T_{timeout}/ t_0 \rceil   + 1$.

\begin{algorithm}
\caption{Approximately Detect Duplicates Flow Tuples using Round-robin Buddy Bloom Filters(DetectDupFlow)}
\label{alg:DetectDupFlow}
\SetKwData{Index}{Index}
\SetAlgoNoLine
\LinesNumbered
\SetKwFunction{RBBF}{RBBF}
\SetKwFunction{BFQuery}{BFQuery}
\SetKwFunction{BFInsert}{BFInsert}
\SetKwFunction{DetectDupNode}{DetectDupNode}
\KwIn{flow key tuples $F=f_1,f_2,\dots,f_n,\dots$}
\KwOut{None}
\BlankLine
$gBBF[0] \dots gBBF[N] \longleftarrow 0$ \\
\For{ $each f_i \in F$}{
    Calculate the index of stale time windows $ind$ \\
    \If { need jump to a new time window }{
        $gBBF[N] \longleftarrow \RBBF{(gBBF[0] \dots gBBF[N-1]),ind}$ \\
    }
    $ (B1,B2) \longleftarrow gBBF[ind]$ \\
    $ \bar{f_i} \longleftarrow $ conjugate flow of $f_i$\\

    \If{ $\BFQuery(f_i,gBBF[N][1]) \neq  True$}{
       \If{ $\BFQuery(f_i,B2) \neq True$ }{
            \eIf{ $\BFQuery(f_i,B1) \neq True$ }{
                $\BFInsert(f_i,B1)$ \\
                $\BFInsert(\bar{f_i},B1)$\\
                \If{$\BFQuery(f_i,gBBF[N][0]) =  True$}{
                    $\BFInsert(f_i,B2)$ \\
                    $\BFInsert(\bar{f_i},B2)$ \\
                    $\DetectDupNode(f_i) $ \\
                }
            }{
                $\BFInsert(f_i,B2)$ \\
                $\BFInsert(\bar{f_i},B2)$ \\
                $\DetectDupNode(f_i) $ \\
            }
       }
    }
}

\end{algorithm}

Algorithm \ref{alg:DetectDupFlow} describes the first stage of NetFlow data processing.
The {\it BFQuery} function is to query element on specified Bloom Filter to check for existence,
it returns true when the element is in the Bloom Filter, otherwise returns False.
The subroutine {\it BFInsert} is to insert element into specified Bloom Filter.
The goal of  Algorithm \ref{alg:DetectDupFlow} is to eliminate flow level duplication.
After the first stage of processing, only those flows seen for the first time are ready for service node discovery.
An end node tuple is denoted by the key combination of \{ ip address, protocol, port number \}.
Thus for a valid flow, two end node tuples are waiting for detection,
source end node tuple and destination one. Any duplicate occurrence of an end node is a candidate of service node.
Representing flows span totally N jumping time windows,
the Round-robin of Buddy Bloom Filters scheme maintains all the duplicate service node tuples summary information,
using $2*(M+1)$ bit arrays. We notate for clarity these bit arrays as two dimensional array $B[2][M+1]$,
that the first dimension can be regarded as a Buddy Bloom Filter, specially $B[2][0]$
can be notate as $B1,B2$ when no confusion exist. $B1,B2$ is used as the current duplicate service node detection.
Each node tuple is checked first.
That is described in Algorithm \ref{alg:DetectDupNode}.


\begin{algorithm}
\caption{Approximately Detect Duplicate service nodes using Round-robin Buddy Bloom Filters(DetectDupNode)}
\label{alg:DetectDupNode}
\SetKwData{Index}{Index}
\SetAlgoNoLine
\LinesNumbered
\SetKwFunction{RBBF}{RBBF}
\SetKwFunction{BFQuery}{BFQuery}
\SetKwFunction{BFInsert}{BFInsert}
\KwIn{flow key tuples $f_i \in F$}
\KwOut{a bit array summary of duplicate service nodes}
\BlankLine
$BBF[0] \dots BBF[M] \leftarrow 0$ \\
\For{ each $ e_i \in (f_i$ associated end node tuples) }{
    Calculate the index of stale time windows $ind$ \\
    \If { need jump to a new time window }{
        $BBF[M] \longleftarrow \RBBF{(BBF[0] \dots BBF[M-1]),ind}$ \\
        Output $B2$
    }
    $ (B1,B2) \longleftarrow BBF[ind]$ \\
    \If{ $\BFQuery(e_i,BBF[M][1]) \neq  True$}{
       \If{ $\BFQuery(e_i, B2) \neq True$ }{
            \eIf{ $\BFQuery(e_i,B1) \neq True$ }{
                $\BFInsert(e_i,B1)$ \\
                \If { $\BFQuery(e_i,BBF[M][0]) = True$ }{
                    $\BFInsert(e_i,B2)$\\
                }
            }{
                $\BFInsert(e_i,B2)$\\
            }
       }
    }
}

\end{algorithm}

Our goal is to find a representative set of duplicate end-nodes.
In practice, when a flow comes, if one of the end-node tuples of flow can be seen in RBBF,
the server side and client side can be determined.
But when either of the service node tuple is seen for the first time, we cannot determinate which side ( source or destination )
is service node when no duplication exists.
Fortunately for us, a typical NetFlow system uses separated process to handle collection and analysis steps.
The collection process is responsible for decoding data from network devices and saving files into filesystem,
and the analysis process is in charge of deriving information like querying ip node or generating statistical information from files.
For efficiency consideration, some analysis tasks are handled by collection process as well,
such as total bytes, average flow size calculation and etc.
This information is stored in file headers for later query or analysis.
So the algorithm proposed here can be used in collection process, 
and before a file is closed and flushed to disk, additional header record of BBF's $B_2$ bit array
can be written to data files.
The structure of a typical saved NetFlow data file is shown in Figure \ref{fig:cbbf-summary-record}.
The summary record written to file is actually the bit array of remembering stage $B2$.
Later the analysis process can read the $B2$ bit array, build up the same hash function $h_1, h_2, \dots, h_k$.
A simple query on the read in bit array can tell whether a end node tuple is already a known service node.


\begin{figure}[htbp]
\begin{center}
  \includegraphics[width=4.0in]{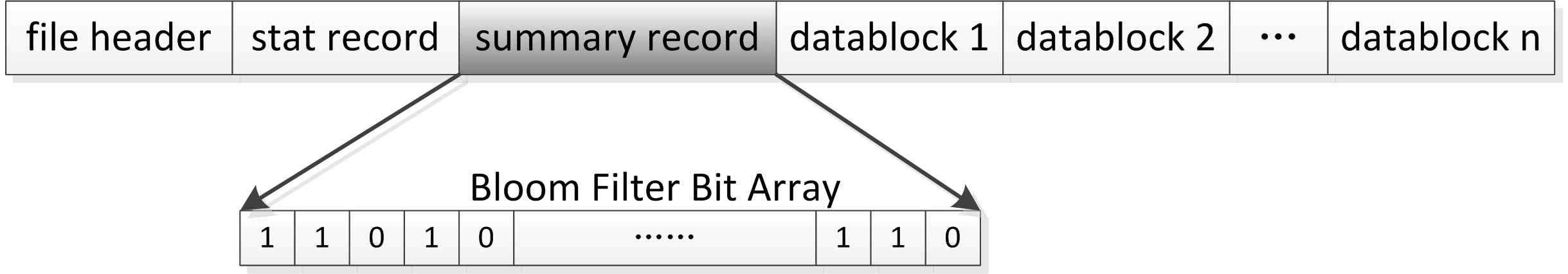}
  \caption{Store Summary Record in Data File}\label{fig:cbbf-summary-record}
\end{center}
\end{figure}

Note there is no limitation of things that can be represented in summary record's bit array.
In practice network operators often dig into raw data files to query some ip addresses.
For example, one wants to check a specific ip address a.b.c.d in daily raw data files.
In the original case, every data file should be opened and reloaded to memory, filtered by the criteria of a.b.c.d .
The matched record is remembered or output. Even the data file contains no records related to a.b.c.d ,
still needs to be checked all through.
The method of summary record can be extended to this case, if we select a Bloom Filter to represent the \{ip address\} sets.
Later query of existence procedure can be largely speeded up.
Only the header information of data files is read first. If there is no match, skip the file; otherwise deal with it like normal operation.
That can dramatically speed up the query operation.


\section{Evaluation and Experiments}

In this section we evaluate the performance of the proposed data structure
and service node detection algorithm with real network data.
We compare the performance of the BBF algorithm with that of a hash table algorithm.
A prototype system based on RBBF and proposed service node detection algorithm is deployed in a large real-life campus network environment.
We also present in this section some use cases to demonstrate how this prototype system
can help security administrators and network operators with their daily tasks.

The prototype system are deployed and experiments are carried out in campus network of Peking University,
which is one of the largest university network in China.
Currently, The capacity of its links connected to internet is 3.7G to IPv4 network, 1G to IPv6 network.
The core and edge routers and switches are all configured to support NetFlow export.
In this paper,
we choose $3$  of them as data sources. CoreR1 is an IPv6-ready router in charge of the Campus' IPv6 traffic routing and forwarding.
CoreR2 is located in the campus' science building and many servers are behind this router. EdegR3 is the edge router of the university,
and all ipv4 traffic to internet is forwarded by this router.
All these flow sources are export NetFlow data using Cisco NetFlow\cite{CiscoNF} version 9.

\subsection{ Evaluation of BBF }


Suppose we want to detect duplication in flows. For IPv6 compatibility, the content of a record is composed
of 128-bit SA(source Address), 128-bit DA(Destination Address), 16-bit SP(Source Port), 16-bit DP(Destination Port),  and 8-bit PROT(Protocol)
$296$ bits in total. To force the alignment of bits in operating system, each flow record actually uses $320$ bits.

A naive approach is to directly track all the flows, i.e., by maintaining a counter for each flow.
Flow and it's conjugate form are deemed as having the same key combinations.
Whenever a counter is larger than one, the flow it belongs to is deemed as a duplicate element.
Usually a hash table is used to store the tuples content and a few auxiliary bytes such as integer counter, etc.
A hash table implementation is an accurate way  to detect duplication. The memory requirement
depends on distinct elements of flows. Figure \ref{fig:uniflows} shows the distinct flows from
these NetFlow data sources.

\begin{figure}[htbp]
\begin{center}
 \subfigure[Distinct Flows]{
    \label{fig:uniflows}
    \includegraphics[width=3.0in]{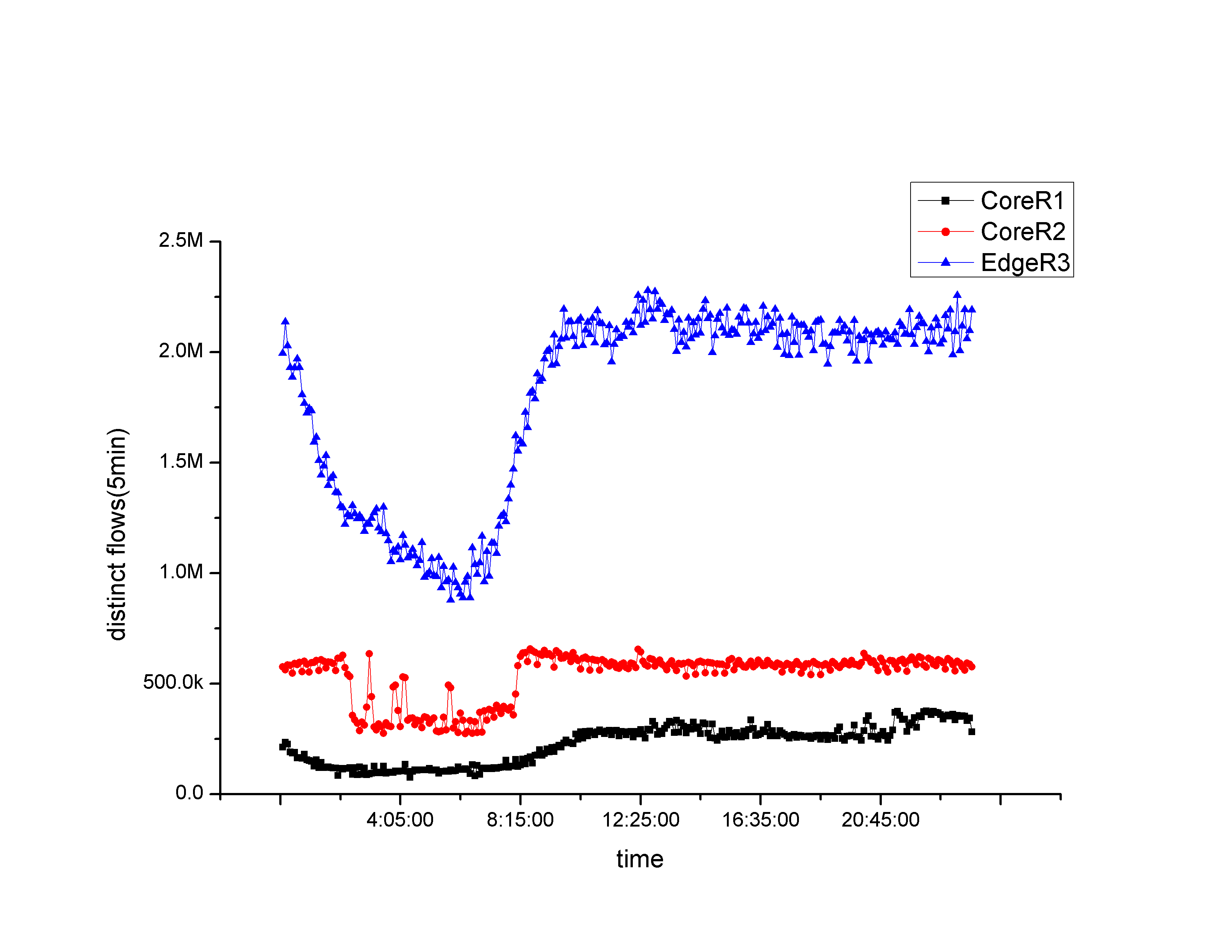}
    }
 \subfigure[Distinct Elements Ratio] {
    \label{fig:r_d}
    \includegraphics[width=3.0in]{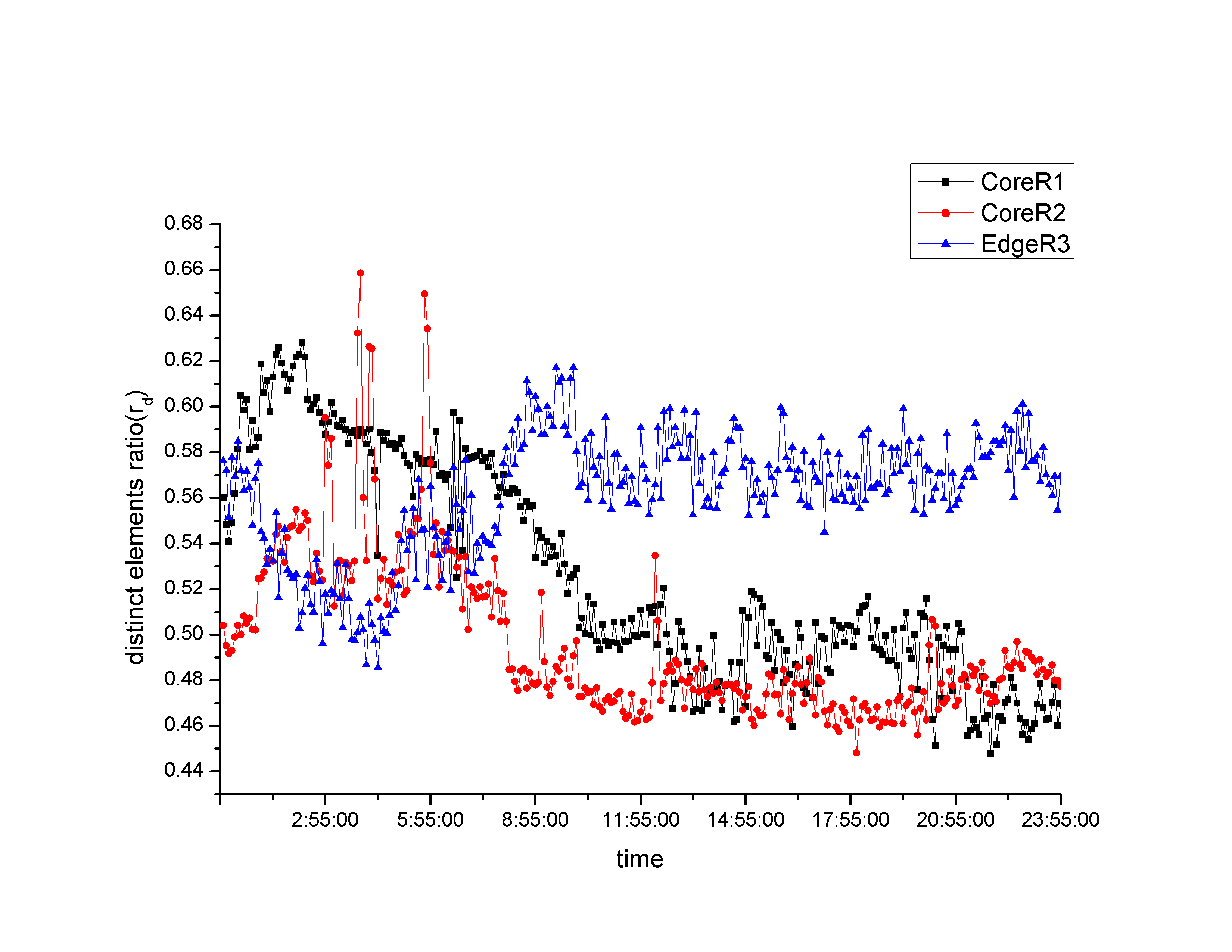}
    }
 \caption{Distinct Elements of Three NetFlow Sources}
\end{center}
\end{figure}

One can simply estimate the memory requirement for hash table is at least $ 2.5M * 320 = 800M $ bits in the worst case.
In practice, hash table cannot be full, otherwise the collision will reduce the performance dramatically.
The space allocated for hash table in this case is about $300M \thicksim 500M$ Bytes, here $1 Byte=8 bits$.

According to equation \ref{eqn:k m opt},  we want to keep the false positive rate below the to threshold $ \varepsilon = 5\%$.
We choose $5$ different hash functions for BBF,\{sax, sdbm, bernstein, elf, fnv \}.
These hash functions are simple and easy to be implemented. For computation efficiency and uniform randomize of hash functions choose, 
some discussion can be found in \cite{hashfunction}.
The optimized number of bit array length of BBF is $ \lceil 2*log_2e*log(1/0.05,2) * nr_d \rceil \approx 14 nr_d$.
Thus memory requirement of BBF is $ 14*2.5M /8 = 4.375M$ Bytes.

Compared to the hash table, the BBF requires less memory space.
Attacks such as DDoS or port/host scan generates large amount of distinct unidirectional flows,
which may soon  fill up the hash table, and cause a lot of collisions in hash table.
This situation has a huge impact on the performance.
But for BBF, more distinct elements than expected increase FP a little without reducing performance.
To avoid the extreme cases, some extra space is allocated for BBF. The preferred parameters
 are listed in Table \ref{tbl:experiment parameters}.

\begin{table}[!h]\label{tbl:experiment parameters}
\tabcolsep 5mm \caption{Parameters of 3 Flow Sources}
\begin{center}
\begin{tabular}{|l|c|c|c|}
\hline { }&{Flows(5mins)}& {$r_d$(avg)} & {m}(Bytes)
\\ \hline
 {CoreR1} & $130K \thicksim 850K$ & $0.45$  &  4M  \\
 {CoreR2} & $500K \thicksim 1.5M$ & $0.49$  &  4M \\
 {EdgeR3} & $1.7M \thicksim 4.5M$ & $0.53$  & 8M \\
\hline
\end{tabular}
\end{center}
\end{table}

We note that the distinct ratio of flows almost stay unchanged in day time. A sharp peak always means
abnormality in network, such as scan, DDoS attack, off-line of servers, etc.
Figure \ref{fig:r_d} shows flows of CoreR2 have some attacks at 2:00AM-6:00AM, which
is confirmed by further inspection as a type of DDoS towards server farm.

\subsection{ Evaluation of RBBF }


Using a single BBF schema, while changing between time windows , boundary effects 
affects the accuracy of flow duplication detection.
If a unidirectional flow is seen at time $t1$ and it's conjugate flow is seen at time $t2$,
we define the time gap as $\delta t = t1-t2 $.
Figure \ref{fig:timespan} shows time gap between bidirectional flows of 3 different NetFlow data sources.
We observe that most time gap of bidirectional flows is around $32ms \thicksim 1024ms$.
Some long lasting bidirectional flows can span to nearly $300s$, which spans the whole single time window .
We use $M$ sequence of BBFs in round-robin schema to eliminate boundary effect.

%

\begin{figure}[htbp]
\begin{center}
 \subfigure[Time Gap of Bidirectional Flows]{
    \label{fig:timespan}
    \includegraphics[width=3.0in]{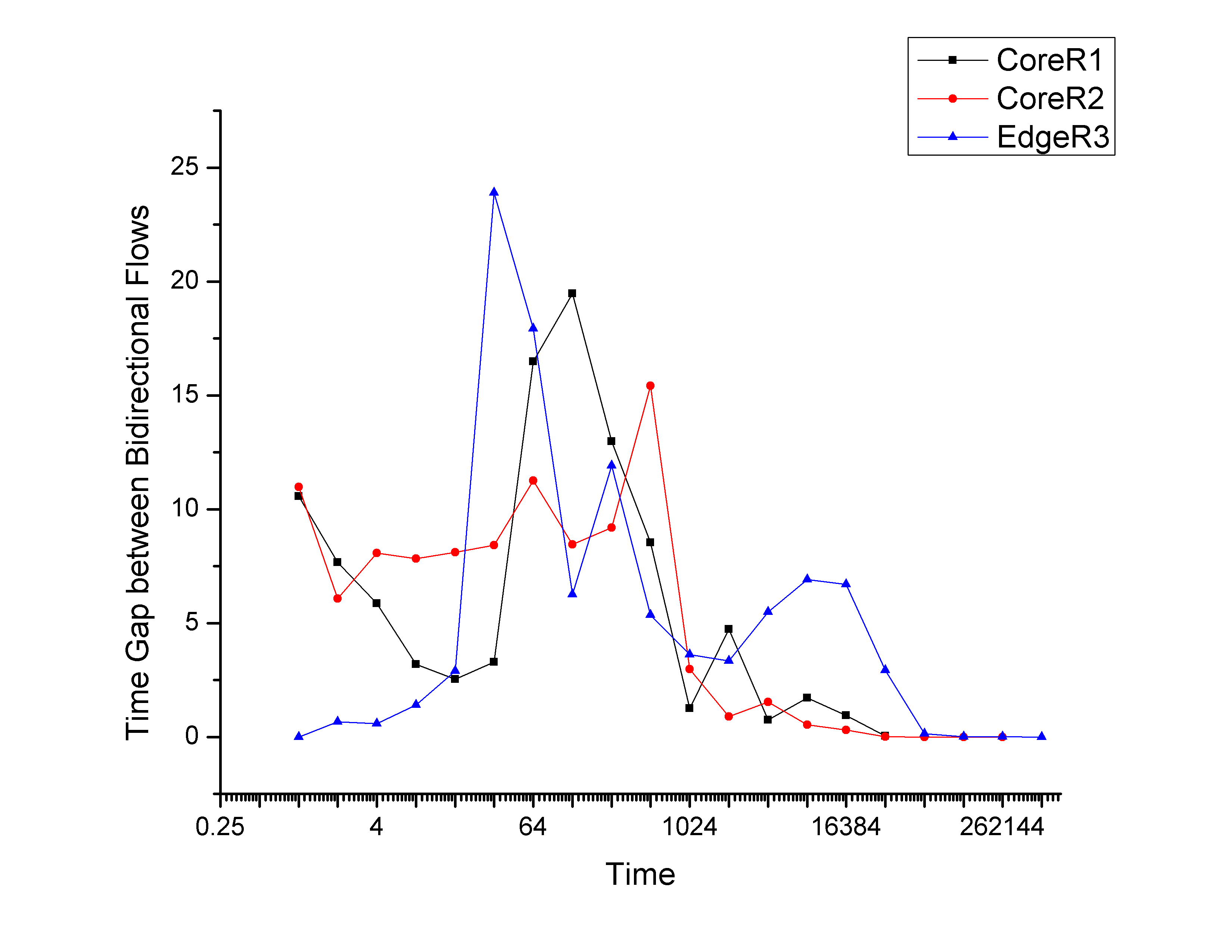}
    }
 \subfigure[Rate of Coincide Nodes] {
    \label{fig:serviceDupHistory}
    \includegraphics[width=3.0in]{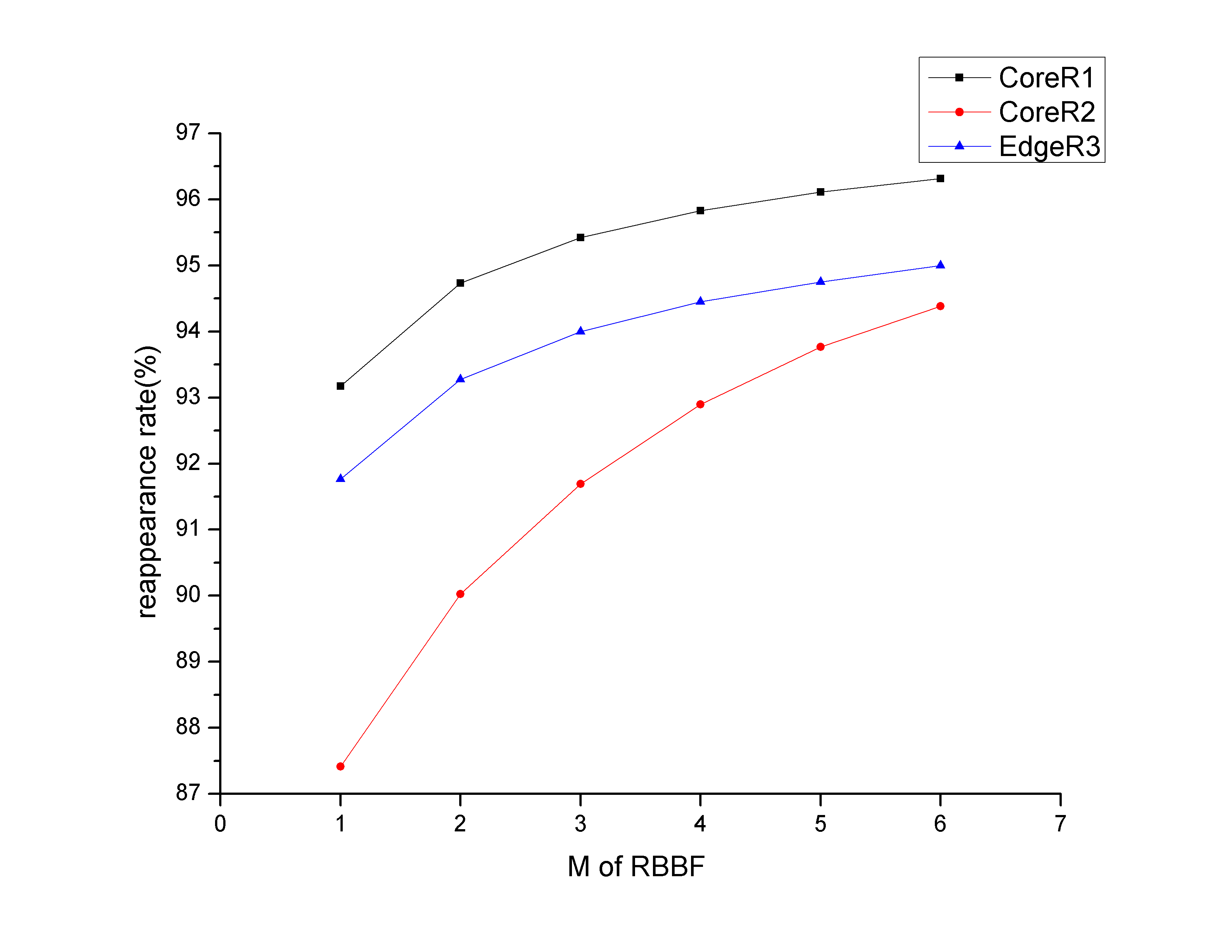}
    }
 \caption{Distinct Elements of Three NetFlow Sources}
\end{center}
\end{figure}


Another benefit by using Round-robin Buddy Bloom Filter is seen from the fact that
most service nodes can be found again after a short period.
For example, a web server provides service at 80/TCP for a long time,
and many different clients connect to this service one after another.
Thus, using history summary BBF can improve the efficiency and  accuracy of identifying server end points.
If a service is already confirmed in history period,
we can certainly deem it as service again when we see it later.
Figure \ref{fig:serviceDupHistory} shows
the coincidence rate between currently detected service nodes and historical ones.

We observe that more BBFs in Round-robin schema result in better
elimination of boundary effect and increase identified service nodes rate,
but surely they result in more memory usage, and increase false positive rate,
for that more distinct elements are repented by summary BBFs.
We choose $M=6$ BBFs in our experiment, partially because
the default timeout active defined in NetFlow specification is 30min\cite{CiscoNF}.

\subsection{The Prototype System: Nfsrv}

We demonstrate the proposed algorithm by our prototype system, Nfsrv.
This system is built on top of nfdump , and the front end is developed as a plugin of Nfsen \cite{nfsen}.

\begin{figure}[htbp]
\begin{center}
  \includegraphics[width=5.0in]{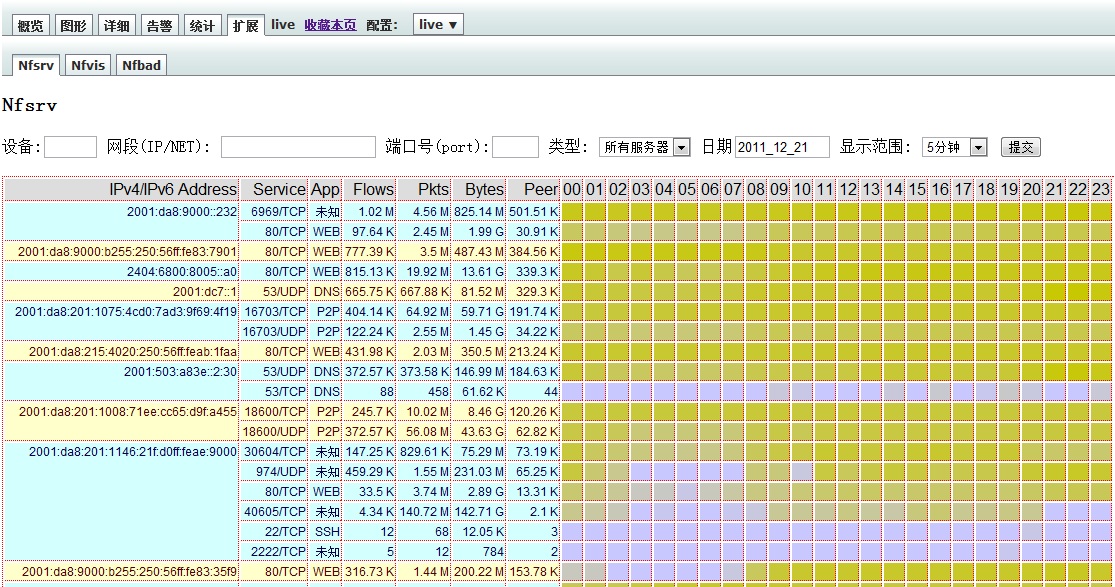}
  \caption{Web-front of Prototype service nodes Discover System: Nfsrv}\label{fig:nfsrv}
\end{center}
\end{figure}

The Figure \ref{fig:nfsrv} shows the front end of the prototype system.
From the system, one can rapidly identify the information of servers.
The proposed algorithm identifies servers that are actively used in the organization network and stores them into a database.
In the web-front of the system,  information about service nodes and some other statistics information is given.
The first column in the table shows the service node's ipv4/ipv6 address, the next column is the transport layer port number.
Other columns and statistics is beyond the discussion of this paper.
For example,the  first row in the table has the IPv6 address of [2001:da8:9000::232], which
can be verified manually as a popular PT sites in CERNET2, bt.neu6.edu.cn. This site provides web access on 80/TCP,
and bt tracker  service on 6969/TCP.
Because there is no assumption of port number in advance, even the dynamically chosen port number can be identified.
As list in the figure, some p2p users also provide service to their peers, using dynamic port number like 16703 or 18600.

The listed service nodes can be used for network monitoring or policy checking.
An example is to track whether a given ip address provide desired services.
If the ip address provides more services than expected,  the web-front of Nfsrv can be used to check
and indicate the compromised using port(s). 

\begin{figure}[htbp]
\begin{center}
  \includegraphics[width=3.0in]{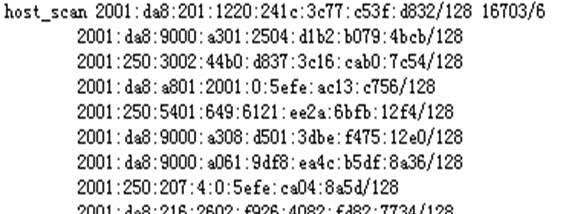}
  \caption{Suspicious Scan Activity}\label{fig:hostscan}
\end{center}
\end{figure}

The prototype system also supports associating service nodes information with  users.
This provides an effective way to  find heavy traffic users and find out what services they are using.
In some cases, hosts are often compromised by a malware that tries to spread.
Scan activity can be detected in such cases. Figure \ref{fig:hostscan} is a demonstration of
suspicious scan activity.
Security administrators can use that information to identify malicious behaviors.

Generally speaking, Nfsrv with service node discovery functionality can provide an overview representation and detail-on-demand capability
for the network administrator to understand what happens in the network. It can be efficiently used to perform forensic tasks.

\subsection{Limitations and Future Work}

As a prototype system of service nodes discovery based on RBBF,
Nfsrv provides a practical network situational awareness solution based on NetFlow flows.
We show some possible use cases where Nfsrv can help network administrators and security operators in their monitoring tasks.
However, Nfsrv has still important limitations that we plan to address in our future work.

Firstly, Nfsrv works with bidirectional flows. In an asymmetric routing environment,
request and reply flows may not pass the same router.
We assume in this case that NetFlow collectors are deployed to cover the flow path bothway.
But when this condition cannot be met, new strategy should be adopted to work with unidirectional flows.
Further more the flow sampling has an evident impact on duplicate detection.
Results from other evaluations of passive detection techniques indicate that
sampling has a limited impact on the overall accuracy.
For example, Bartlett. etc \cite{Bartlett2007}
reports that capturing only 16\% of the data results only in an 11\% drop in discovered servers.
However, the consequence of flow level random sampling shall be further studied.

The current algorithm use duplication information in NetFlow data flows to detect service nodes.
But if two end nodes  are both identified as duplicate elements,
additional information is needed such as the cardinal of connections of each end nodes.
One example is that busy client nodes, such as NAT boxes or fast scanning nodes,
may reuse the same ephemeral port in a short time,
which may cause such false positive error of server detection.
Specially ,some malicious software may imitate server behavior by repeatedly
using the same source port number for scanning or communication.
An effective approach other than naive counting method need to be further studied for these cases.

We also admit that Nfsrv works at the network layer and therefore heavily relies on port numbers.
Nfsrv can tag out service nodes using dynamic port numbers,
but it can be hard for a network operator to identify the application behind a service only by transport layer information.
This issue arises from the fact that some applications use random ports or hide behind well known ports.
For example some P2P applications use port 80 or port 443, which are normally reserved for web traffic, in order to evade firewall protection.
Flow-based traffic classification proves that it is possible to accurately identify applications using only NetFlow.
We currently work on developing traffic classification method based on identification of host behaviors and relationships between hosts.
We believe that discovering community of interest patterns between hosts would be an effective step to identify not
only applications type but also unveil communication style behind P2P network or botnets.


Finally, the current prototype system has limited ability of intrusion detection.
We plan to integrate some machine-learning approach to build up a more useful system
for network administrators and security operators .

\section{Conclusion}


In this paper, we propose the Buddy Bloom Filter data structure to select duplicate elements
from data stream, especially network traffic data of NetFlow.
Multiple Buddy Bloom Filters with Round-robin schema are used
to allow outdated information to be removed and mitigate the boundary effect.
To take the advantage that bitwise operation can be applied to Buddy Bloom Filter structure,
we also suggest using a summary BBF to speed up the query process on history BBFs.
The properties and analysis of the RBBF are given.

A two stage algorithm specially for discovering service nodes from NetFlow data is proposed.
Firstly the flow level duplication is eliminated and then both end node tuples
are checked for duplication as service node candidates.
We formally and experimentally study the properties and performance of our algorithm.
We empirically evaluate it and report the conditions under which our method outperforms the alternative
methods, in particular hash table method, which needs extra space linear
in the buffer size to obtain constant or nearly constant processing time.
Timely information on what is occurring in their networks is crucial for network and security administrators.
RBBF based service nodes discovery provides a near real-time report on network activities.

In this paper, we describe the architecture and algorithm using RBBF to
discover service nodes by duplicate detection from NetFlow data,
and demonstrate a number of real use-cases using our prototype system.
Our future work includes development and evaluation of more accurate service discovery methods.
We also plan to add some more intrusion detection rules based on host behavior
and add some machine learning strategy to make the system more useful.

\section{Acknowledgements}

The first author would like to thank Lv Jie  at Computer Center of Peking University
for their detailed comments on this paper.
This work was partially supported by National Key Basic Research Development Program (973 Program)
under Grant No.$2009CB320505$.

\bibliographystyle{latex8}
\bibliography{latex8}

\begin{large}
\textbf{Biographies}
\end{large}

\includegraphics[width=1in,clip,keepaspectratio]{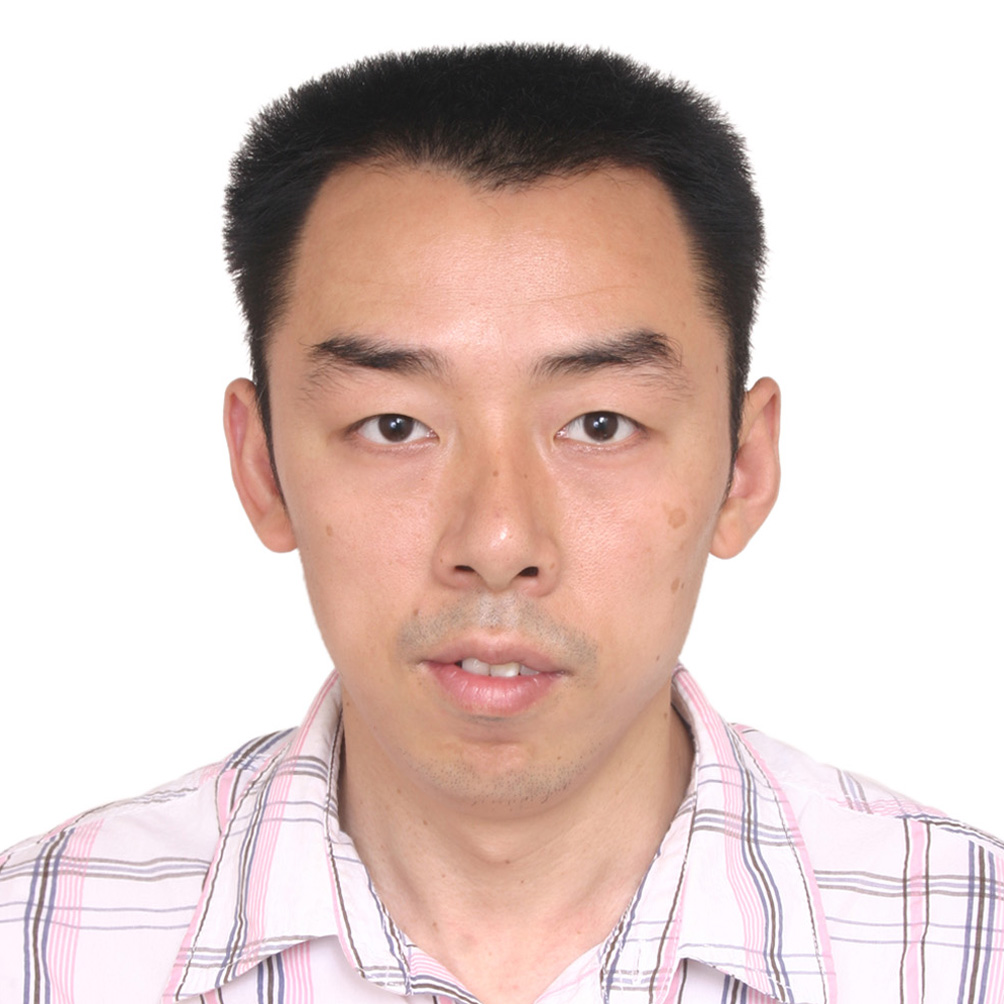}
{\noindent{\bf Zhou Changling}, received his B.S. and M.S. from Peking University.
Currently, he is working at Computer Center of Peking University. He is also a
Ph.D candidate of School of Electronics Engineering and Computer Science of Peking University.
His current research interests include network traffic analysis,
wireless and network management. Email: zclfly@pku.edu.cn.}

\vspace{5ex}

\includegraphics[width=1in,clip,keepaspectratio]{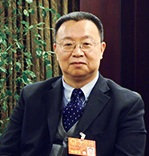}
{\noindent{\bf Xiao Jianguo}, Professor at Peking University.
Prof. Xiao is currently engaged in the research on
image processing, text mining and
web information processing .}

\vspace{10ex}

\includegraphics[width=1in,clip,keepaspectratio]{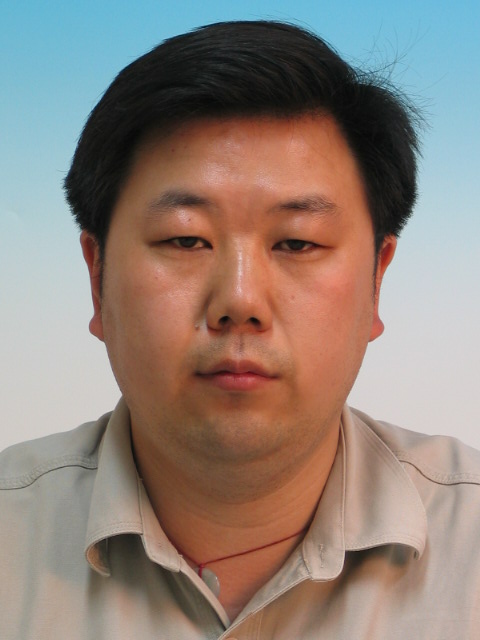}
{\noindent{\bf Cui Jian}, Associate Professor at Peking University.
His current research interests include on
network management, network architecture,
security .}

\vspace{15ex}

\includegraphics[width=1in,clip,keepaspectratio]{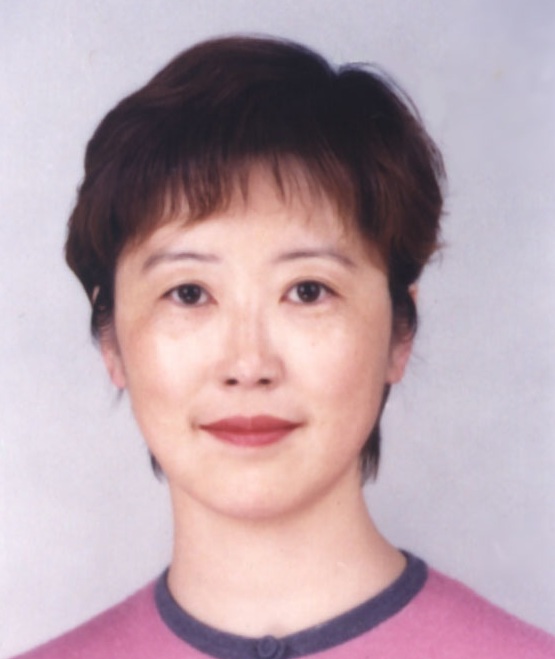}
{\noindent{\bf Zhang Bei}, Professor at Peking University.
Prof. Zhang research interests includes
network architecture, network management,
security, information management.}

\vspace{13ex}

\includegraphics[width=1in,clip,keepaspectratio]{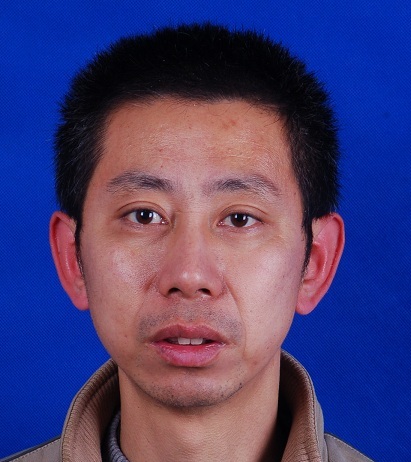}
{\noindent{\bf Li Feng}, Associate professor at Information Technology Center,
SHENYANG Academy of Governance. His current research interests include network architecture,
network management, network application }

\end{document}